\documentclass[sigconf]{acmart}
\AtBeginDocument{%
  \providecommand\BibTeX{{%
    \normalfont B\kern-0.5em{\scshape i\kern-0.25em b}\kern-0.8em\TeX}}}

\setcopyright{acmcopyright}
\copyrightyear{2022}
\acmYear{2022}
\acmDOI{10.1145/1122445.1122456}

\acmISBN{978-1-4503-XXXX-X/18/06}

\usepackage{balance} 
\usepackage{times}
\usepackage{soul}
\usepackage{url}
\usepackage[utf8]{inputenc}
\usepackage[font=small]{caption}
\usepackage{graphicx}
\usepackage{amsmath}
\usepackage{amsthm}
\usepackage{booktabs}
\usepackage{algorithm}
\usepackage{amssymb}
\usepackage{xcolor}
\usepackage{soul}
\usepackage[utf8]{inputenc}
\usepackage{subcaption}
\usepackage{url}
\let\savedegree\bigtimes
\let\bigtimes\relax
\usepackage{mathabx}
\let\bigtimes\savedegree
\usepackage{multirow}
\usepackage{cases}
\usepackage{algpseudocode}
\usepackage{verbatim}
\usepackage{enumitem}
\usepackage{mathrsfs}
\usepackage{tikz}
\theoremstyle{plain}

\DeclareMathOperator*{\argmax}{argmax}
\DeclareMathOperator*{\argmin}{argmin}
\newcommand{\Access}{\mathsf{AC}}

\newcommand{\Anti}{\overline{N}}

\newcommand{\boomerang}{\mathsf{bdeg}}
\newcommand{\bndeg}{\mathsf{bnde}}
\newcommand{\brandom}{\mathsf{bran}}

\newcommand{\degr}{\mathsf{xdeg}}
\renewcommand{\degree}{\mathrm{deg}}
\newcommand{\degN}{\mathrm{ndeg}}

\newcommand{\dist}{\mathsf{dist}}

\newcommand{\Ext}{\mathrm{X}}

\newcommand{\g}{\mathfrak{G}}

\newcommand{\local}{\mathsf{ldeg}}
\newcommand{\lrand}{\mathsf{lran}}

\newcommand{\N}{\mathbb{N}}

\newcommand{\prowl}{\mathsf{prowl}}

\newcommand{\random}{\mathsf{xran}}

\begin{document}


\title[Opinion convergence]{Centralization Problem for Opinion Convergence in Decentralized Networks}



\author{Yiping Liu}
\affiliation{
  \institution{Univsersity of Auckland}
  \country{New Zealand}}
\email{	yliu823@aucklanduni.ac.nz}

\author{Jiamou Liu}
\affiliation{
  \institution{University of Auckland}
  \country{New Zealand}}
\email{jiamou.liu@auckland.ac.nz}

\author{Bakhadyr Khoussaino}
\affiliation{
  \institution{University of Electronic Science and Technology of China}
  \country{China}}
\email{bmk@cs.auckland.ac.nz}

\author{Miao Qiao}
\affiliation{
  \institution{University of Auckland}
  \country{New Zealand}}
\email{miao.qiao@auckland.ac.nz}

\author{Bo Yan}
\affiliation{
  \institution{Beijing Institute of Technology}
  \country{China}}
\email{yanbo@bit.edu.cn}
\renewcommand{\shortauthors}{Anonymous, et al.}

\begin{abstract}
This paper aims to provide a new perspective on the interplay between decentralization -- a prevalent character of multi-agent systems -- and centralization, i.e., the task of imposing central control to meet system-level goals. In particular, in the context of networked opinion dynamic model, the paper proposes and discusses a framework for centralization. More precisely, a decentralized network consists of autonomous agents and their social structure that is unknown and dynamic. Centralization is a process of appointing agents in the network to act as access units who provide information and exert influence over their local surroundings. We discuss centralization for the DeGroot model of opinion dynamics, aiming to enforce  opinion convergence using the minimum number of access units. We show that the key to the centralization process lies in selecting access units so that they form a dominating set. We then propose algorithms under a new local algorithmic framework, namely prowling, to accomplish this task. To validate our algorithm,  we perform systematic experiments over both real-world and synthetic networks and verify that our algorithm outperforms benchmarks.
\end{abstract}



\maketitle

\section{Introduction}
{\em Opinion dynamics} aims to illustrate how opinions of individuals evolve in a crowd  \cite{de2018shaping}, so as to provide insights on important challenges in multi-agent systems such as group decision making  \cite{dong2018consensus}.
In an opinion dynamics model, every agent holds an opinion towards a certain issue.
The agents interact by revealing their opinions to others, and in turn updating their opinions as a result of such interactions. Such a model is {\em decentralized} in the following sense. 
First, the agents' states, which capture not only their opinions but also with whom they interact, evolve in a way without the control of a higher-level entity, i.e., the ``system controller''.
Second, in the presence of a large number of agents, any potential system controller has  only {\em partial information} about these states, meaning that agents' opinions and their social interactions are not directly accessible.
In such a decentralized system, however, the agents may fail to reach a consensus, a key objective of the opinion dynamic process. Therefore a need arises to impose some coordinated control.

The interplay between decentralized systems and centralized control has been the subject of intense studies, with implications in fields like networking, political science, and information systems  \cite{galloway2004protocol,everton2015dark,wangDecentralizedCommunicationControl2015}.
Going along with this theme, we formulate  {\em decentralized networks} to capture the type of opinion dynamics models described above. In short, a decentralized network consists of agents whose repeated interactions form a social network, while unifying the following properties.
(1) {\bf Dynamic topology}: Here, the social network is not static but rather it changes with time. There has been a growing interest in dynamic networks as well as how changes to the network topology affect opinions \citep{singh2012accelerating,marchant2015manipulating}.
(2) {\bf Unknown topology}: The social networks are not provided upfront due to their scale and the locality of knowledge. The study of these types of networks requires
discovering the network topology through some form of exploration \citep{zhuang2013influence,wilder2018maximizing,he2019opinion,marchant2017convention}.
Over a decentralized network, we define a process of {\em centralization} that imposes central control via two functionalities: {\bf monitoring}, meaning that a number of agents are appointed as {\em access units} to gather information about a local region; and {\bf influencing}, meaning that these access units may exert influence to the agents as directed by the system controller.
This setup captures two extrema: (1) when no agent is an access unit, which means that the system is {\em fully decentralized}, and (2) when all agents are access units, which means that the system is {\em fully centralized}. Centralization refers to the process of appointing access units, starting from extrema (1), moving towards extrema (2) until an appropriate level of  control is imposed. The  challenge is to adopt suitable strategies to explore the unknown and dynamic social network and appoint access units along the way. 

This paper focuses on the centralization problem for opinion convergence. We propose a decentralized version of the well-established {\em DeGroot model} to drive agents' opinions \cite{degroot1974reaching}: The opinion of an agent is a real value in $[0,1]$, which is updated by the average opinion of neighboring agents. The goal is to have the opinions of all agents converged to 1. One paradigm to achieve this, as proposed in \cite{singh2012accelerating}, is to deploy ``committed agents'' whose opinion is fixed at 1. We model these committed agents as the access units in our centralization processes. We now summarize our contributions:

\smallskip
\noindent
{\bf (1)}. We define decentralized networks and  formulate the centralization problem over decentralized DeGroot model. These are presented in Sections~\ref{S:Basics} and \ref{S:Statement}. These notions can potentially be defined for other tasks such as influence maximization or norm emergence  \cite{yadav2016using,marchant2017convention}, thereby triggering fruitful new research.

\smallskip
\noindent
{\bf (2)}. We seek
to effectively facilitate opinion convergence using the minimum number of access units.
We establish a connection between finding a dominating set, an important concept in network theory, and the centralization  problem for opinion convergence. We provide theoretical results and algorithms in Sec.~\ref{S:Statement} and \ref{S:Prowling}. These form the backbone of our experimental results in the last section.

\smallskip
\noindent
{\bf (3)}. We explore {\em local algorithms}  \cite{brautbar2010local} to solve the centralization problem. For this, we propose a new algorithmic framework, {\em prowling}, that probes an unknown network while selecting access units. We design three prowling algorithms: one {\em exterior-based}, and two {\em B-set} prowling algorithms (see Section~\ref{S:Prowling}). These algorithms ensure opinion convergence under general conditions. Through empirical validation over dynamic real-world and synthetic  datasets, we validate the effectiveness of our algorithms (See Section~\ref{sec:experiment}).

\section{Related work}\label{sec:related}
This paper is related to work on {\em local algorithms}  which aims to identify global properties of a graph using only local information \citep{brautbar2010local}. For example, \citep{borgs2012power} proposes an algorithm that explores an unknown network by repeating the following: greedily select  a ``known'' node, and then  make a random move to an unknown neighbor of the selected node. This method guarantees to cover a (static) graph using a small set of nodes that has a $O(\log(\Delta))$-factor from an optimal set, where $\Delta$ is the maximum degree. If we view nodes traversed by the local algorithm as access units, then this work fits the framework studied in our paper.
Wilder et al. utilize local algorithms for influence maximization \cite{wilder2018maximizing}, where they design a {\em probing} mechanism to collect information, i.e., walks in the unknown part of the network. The structural statistics collected by repeated probes direct the selection of seed nodes, which can also be considered as access units. In our paper, we adopt a different probing method. Instead of using probes as  means to sample information, we implement a search procedure through the unknown part of the network to locate the next access unit. We call such types of algorithms {\em prowling} and consider this procedure as a less costly and adaptive alternative to the exploration strategy used in \cite{wilder2018maximizing}.

{\em Opinion dynamics}, with its application to group decision making, has been intensively studied  \cite{abrahamsson2019opinion,ghaderi2014opinion,dietrich2017control,dong2018consensus}.
The vast majority of work on opinion dynamics assumes that the social network is given as input,  implying that the input is both static and fully-known. This allows us to directly employ algorithms from graph theory and combinatorial optimization. Singh et al. \cite{singh2012accelerating} study the co-evolution of network structures and individual opinion and found that committed agents can accelerate opinion convergence.
 He et al. \cite{he2019opinion} investigate opinion maximization by selecting seed agents under the assumption of partial information  on edge weights.
Their work, however,  omits  information collection,  which is  important in  real-world applications.

While no work exists that investigates opinion dynamics in  unknown and dynamic social networks, several papers study the related topic of  {\em influence maximization} in unknown networks. Social influence is normally modeled using a cascading model, giving rise to a submodular function. 
When the social network is not available, it is common to assume that sampled influence pairs are available to extract information regarding the social network \cite{balkanski2017limitations,balkanski2017importance,yan2019maximization}. 
These  work invokes submodular function optimization, which is not applicable in models of opinion dynamics.   Wu et al. \cite{wu2019maximizing} assume that the network topology is given while the edge weight is unknown in a cascade model. Zhuang et al. \cite{zhuang2013influence} study networks that are initially known, but the structural changes are only observable by probing. Their method first selects a set of probing nodes that reveal their neighborhood topology and then selects a set of seed nodes accordingly. Our setting is somewhat similar to \cite{zhuang2013influence} except that our network is not initially known, and we consider opinion dynamics rather than information diffusion.

Another aspect of our model, which is distinctive from all work we mentioned above, lies in selecting access units as an iterative process as the system evolves, i.e., centralization is {\em adaptive}. This resembles the setting of {\em adaptive influence maximization} that iteratively selects seed nodes while observing network evolution \cite{yadav2016using}.

\section{Decentralized Network}\label{S:Basics}

\subsection{The general setup}
We define our general framework of {\em decentralized networks}, then we present an instance of this framework for opinion dynamics modeling. Assume that there is a fixed (possibly infinite) universal set of {\em agents} $U$. These agents will be nodes of our social network.
Let $Q$ be a set of {\em states} of the agents. 
For any subset $V\subseteq U$, use $[V]^2$ to denote the set of all pairs $\{x,y\}\subseteq V$ with $x\neq y$.

\begin{definition}\label{def:network}
A {\em network instance} is a triple $G=(E,C,S)$ where $E$ is a finite set of undirected edges $E\subseteq [U]^2$, $C\colon U\to Q$ assigns every agent in $U$ a state, and $S\subseteq U$ is a set of {\em access units}.
\end{definition}
  Since we assume that $U$ is the de facto node set throughout the paper, Definition~\ref{def:network} does not explicitly include a node set in the graph. In the next definition, we represent time as  $t=0,1,2,\ldots$
\begin{definition}
A {\em decentralized network} is a triple $$\g=\left(\left\{e(t)\mid e\in [U]^2\right\}_{t\in \N}, \{v(t)\mid v\in U\}_{t\in \N},\{S_t\}_{t\in \N}\right)$$
where $e(t)$ is a random variable with domain $\{0,1\}$, $v(t)$ is a random variable with domain $Q$, and $S_t\subseteq U$ is
a set of access units.
\end{definition}
Intuitively, for $e=\{u,v\}\in [U]^2$, $e(t)=1$ iff there is an $\{u,v\}$ at time $t$, and for $v\in U$, $v(t)$ is the state of $v$ at time $t$.  The {\em network instance at time $t$} is $G_t=(E_t,C_t,S_t)$, where $E_t=\{e\in [U]^2\mid e(t)=1\}$ and $C_t(v)=v(t)$. We will therefore refer to $\g$ also by the sequence $\{G_t\}_{t\in \N}$.
We assume that the sequences $\{E_t\}_{t\in \N}$ and $\{C_t\}_{t\in \N}$ are the outcomes of Markovian stochastic processes with known distributions $\{P_e\}_{e\in [U]^2}$ and $\{P_v\}_{v\in V}$ where
any $e(t+1)\sim P_e(\cdot \mid E_t,C_t,S_t)$ and any $v(t+1)\sim P_v(\cdot \mid E_t,C_t,S_t)$.  
Note also that this stochastic definition subsumes {\em deterministic} state transitions, i.e., function  $f_v$  that produces  $v(t+1)$  given $E_t$ and $C_t$. Indeed, our opinion dynamic model will adopt deterministic state transitions (see Sec.~\ref{DeGroot}).  Note that the sequence $\{S_t\}_{t\in \N}$ is not a stochastic process but rather a deterministic sequence of subsets of agents. The algorithms that we describe  will ensure that $S_t \subseteq S_{t+1}$ for all $t\in \N$. 
The system knows the dynamics of the agents. 

 The access units $S_t$ represent the means for a system controller to monitor and influence the agents: For any $v\in S_t$, the {\em accessible region} of $v$ is a set $\Access_t(v)\subseteq U$ that contains $v$.
The {\em accessible region} of set $S$ is $\Access_t(S)=\bigcup_{v\in S} \Access_t(v)$. The accessible region plays two roles: (1) $\Access_t(v)$ defines $v$'s scope of influence. In other words, for any $S,S'\subseteq U$, for any node $u\in U$, $P_u(\cdot\mid E_t,C_t,S)=P_u(\cdot \mid E_t,C_t,S')$ whenever $\{v\in S\mid u\in \Access_t(v)\}=\{v\in S'\mid u\in \Access_t(v)\}$. (2) $\Access_t(S_t)$ bounds the information accessible by the system controller.  In particular, {\em only nodes in $S_t\cup \Access_t(S_t)$, their states, as well as all edges attached to $\Access_t(S_t)$ are known by the system controller}; any other information about the network is not accessible, nor influenced by the system controller.   Note that, with these assumptions, the system controller can also tell whether  there is a node $v\notin \Access_t(S_t)$ that is connected to a node in $\Access_t(S_t)$.

\subsection{Decentralized DeGroot model}\label{DeGroot}
In the rest of this paper, the decentralized network will be a decentralized DeGroot model, which we define in this section.  
Assume that at any time $t$, the network instance $G_t$ consists of only one (finite) non-trivial connected component $V_t\subseteq U$ such that all edges and access units are from this connected component. We restrict ourselves to {\em incremental networks} where any edge, once appeared, will not be deleted, and thus $V_t \subseteq V_{t+1}$. The opinion of any node $v\in U\setminus V_t$ is irrelevant (therefore $v(t)$ can be assumed to be 0 for simplicity). So, $V_t$ represents the agents that {\em exist} at time $t$ and we may write the network instance $G_t$ as the tuple $(V_t, E_t, C_t, S_t)$.  We assume that the edges evolve independently of access units, i.e., $P_e(\cdot\mid E_t,C_t, S)=P_e(\cdot \mid E_t,C_t, S')$ for all $S$ and $S'$.

Write $\dist_t(u,v)$ for the length of a shortest path in $E_t$ for $u,v\in V_t$. The {\em accessible region} of node $v$ is parametrized by a {\em radius} $r$:
\begin{equation}\label{eqn:access}
\Access^{(r)}_t(v)=\{u\in U\mid \dist_t(u,v)\leq r\}.
\end{equation}
As $r$ is fixed, we write $\Access_t(\cdot)$ for $\Access^{(r)}_t(\cdot)$. In principle, $r$ can be set to any positive integer, we usually assume that an access unit can only influence its close proximity and thus $r\leq 2$.

Following DeGroot model \cite{degroot1974reaching}, an opinion is a value in $Q=[0,1]$, and agents align their opinion with their neighbors. The system controller imposes control as follows: (1) Each $u\in S_t$ is a {\em ``committed agent''} with a fixed opinion at value 1. (2) Every $v\in \Access_t(S_t)$ is assigned an ``accessor'' $u\in S_t$, where $v\in \Access_t(u)$. This accessor has direct influence on  $v$ in the sense that $v$'s opinion at time $t+1$ is the average of the opinion values of its non-$S_t$ neighbors at time $t$ and the opinion of $u$ fixed at 1. Formally, we define \footnote{This definition resembles the DeGroot model with committed agents \cite{liu2014control,abrahamsson2019opinion}. The only difference is that here the effect on $v(t+1)$ of multiple access units is identical to the effect of a single unit, whereas the original DeGroot model treats all agents equally.}:

\begin{definition}The {\em decentralized DeGroot model} is a decentralized network where the state transition function of nodes is set as:
\begin{equation}\label{eqn:degroot}
v(t+1)=\begin{cases}
1 & \text{if $v\in S_t$},\\
\sum_{u\mid \{u,v\}\in E_t} \frac{u(t)}{|\{u\mid \{u,v\}\in E_t\}|} & \text{if $v\notin \Access_t(S_t)$},\\
\frac{\left(1+\sum_{u\mid \{u,v\}\in E_t,u\notin S_t} u(t)\right)}{1+|\{u\mid \{u,v\}\in E_t,u\notin S_t\}|} & \text{otherwise.}
\end{cases}
\end{equation}
\end{definition}

\section{The Centralization Problem}\label{S:Statement}


\subsection{Problem formulation}
{\em Centralization} is a computational process performed by the system controller that selects access units in a decentralized network. 
As the network is initially unknown, selecting an access unit requires a form of ``information gathering''.   
Starting with an initial network instance $(V_0,E_0,C_0,S_0)$, where $S_0=\varnothing$, the process repeatedly executes a {\em local algorithm} that
selects access units. 
When executed at time $t>0$, the algorithm uses information provided in $\Access_t(S_{t-1})$ (states and edges of nodes) and
returns a node $u_t\notin S_{t-1}$, which is the next access unit, i.e., $S_{t}=S_{t-1}\cup \{u_t\}$.

We now propose a type of local algorithms called {\em prowling}. Let the {\em closed neighborhood} of $v\in V_t$ be $N_t(v)=\{u\in V_t\mid \dist(u,v)\leq 1\}$. The {\em $1^+$-neighborhood} $N^+_t(v)$ of  node $v$ is a data structure containing the states of all $v\in N_t(v)$ and their degrees in $G_t$.
\begin{definition}
A  ({\em $k$-step}) {\em prowling algorithm} $\prowl(S_{t-1})$, given the set $S_{t-1}$, implements a search strategy at time $t$:
\begin{itemize}[leftmargin=*]
\item Start at  an agent $w_0$ from  $\Access_{t}(S_{t-1})$.
\item For $0\leq i<k$, if $w_i\in \Access_t(S_{t-1})$, move from $w_i$ to a neighbor $w_{i+1}$ of $w_i$; if $w_i\notin \Access_t(S_{t-1})$, then first inquire about the $1^+$-neighborhood $N^+_t(w_i)$ and then move from $w_i$ to a neighbor $w_{i+1}$ of $w_i$ based on the obtained information. 
\item The last agent $w_k$ is chosen as output.
\end{itemize}
\end{definition}
The prowling algorithm resembles a {\em name-generation process} used in social network analysis \cite{chung2005exploring}, where an interviewer asks a participant (an agent in the network) to reveal the names of those who came into direct contact with the participant, then those whose names are revealed are asked the same question and so on. Note that every step of the prowling algorithm can make use of the degree of the current node.  The length $k$ is a fixed parameter. Figure~\ref{fig:prowl} illustrates a run of a 4-step prowling algorithm. 

\begin{figure}
\resizebox{!}{3cm}{\includegraphics{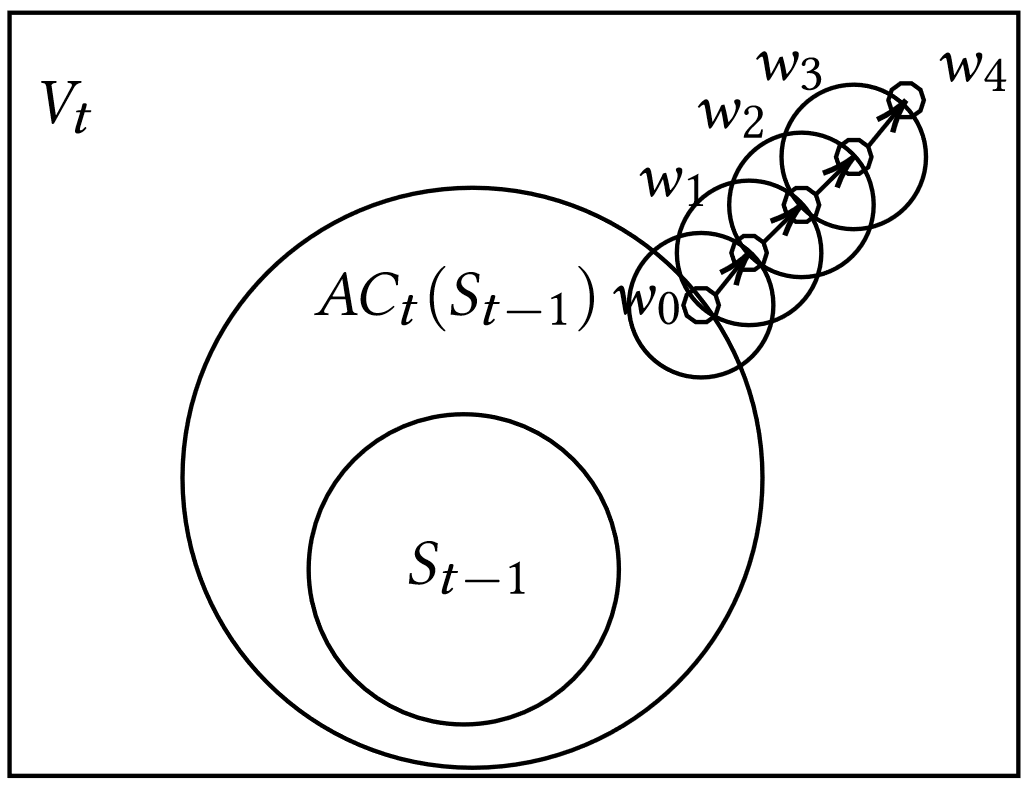}}
\caption{An execution of prowling: Starting from $w_0$ in the accessible region, the 4-steps prowling produces the node $w_4$ as next access unit.}\label{fig:prowl}
\end{figure}

Initially the  opinion values of all agents are $0$: \  $v(0)=0$ for $v\in U$. The desired outcome of a centralization process studied in this paper is when the opinions of all agents in $V_t$ converge to 1. 
\begin{definition}
We say that the {\em opinions converge} in a decentralized network $\g$ if  $\forall \varepsilon>0$, $\exists t>0$, such that
\begin{equation}\label{eqn:converge}
\forall v\in V_t\colon v(t)\geq 1-\varepsilon.
\end{equation}
In this case, the {\em $(1-\varepsilon)$-centralization cost} of $\g$ is $|S_t|$ where $t$ is the first time when \eqref{eqn:converge} holds.
\end{definition}
The {\em network centralization problem} asks for a prowling algorithm 
to explore a decentralized network $\g$ and select access units that achieve
opinion convergence with the smallest centralization cost.

\paragraph*{\bf Remark.} We now elaborate on the motivation behind the adoption of prowling, rather than other forms of local algorithms, in the network centralization problem: Firstly, the local algorithms as proposed in \citep{borgs2012power} explore a network without going beyond the ``known region''. This means that at every step, the selected access unit will always be a member of $\Access_t(S_{t-1})$. While such approaches strictly adhere to ``accessing only local information'', over large networks, too many access units may be chosen for our purpose.  Secondly, the type of probing algorithms proposed in \cite{wilder2018maximizing} were used to sample paths in an unknown network. 
In this way, the next access unit may not be part of the currently accessible region. However, such a method may incur a high cost when a lot of repeated sampling is used to select one access unit. Moreover, this algorithm was designed for static networks assuming that the network structure stays stable between repeated sampling.  Prowling, on the other hand, strikes a balance between the two types of local algorithms above. It allows the selection of nodes outside of the accessible region {\em as the network evolves}, while not incurring the high cost with repeated sampling.

\subsection{From stabilized to unit-growth networks}\label{sec:unit}
To study the network centralization problem, we first consider the case when the network {\em stabilizes} at some time $\ell>0$, i.e., that $e(t)=e(\ell)$ for all $e\in [U]^2$ and $t>\ell$.
\begin{theorem}[Convergence theorem]\label{Thm:converge} If a decentralized network stabilizes at a time $\ell$ with $S_\ell\neq \varnothing$, then opinions converge. 
\end{theorem}
\begin{proof}
Let $\iota=\min\{t\mid  S_t\neq \varnothing\}$.
For convenience and w.l.o.g., we prove the case when $\iota=\ell$, $S_\iota=\{s\}$ for some $s\in V_\iota$ and $v(\iota)=0$ for all $v\in V_\iota\setminus\{s\}$.
We show that $\forall v\in V_\iota \forall t\geq \iota$, $v(t+1)\geq v(t)$.
Indeed, for any $v\in V_\iota$, let $t_v=\min_{t}\{t\in \N \mid v(t+1)< v(t),t\geq \iota\}$ (where $\min_{t} \varnothing=\infty$).
Pick $w$ from $\argmin_v \{t_v\mid v\in V_\iota\setminus\{s\}\}$. Suppose $t_w<\infty$.
Note that $t_w>\iota$, as clearly for any $v\in V_\iota\setminus\{s\}$, $v(\iota+1)\geq v(\iota)$.
Since $w(t_w+1)< w(t_w)$ and our model is a DeGroot model, there is a neighbor $x\in N_{t_w}(w)$ such that $x(t_w)<x(t_w-1)$. This means that $t(x)<t(w)$, contradicting the def. of $w$.

Thus, the opinion of each $v\in V_\ell$ is a monotonic sequence $v(\ell)$, $v(\ell+1)$, $\ldots$ bounded by  1. Let $v^*$ be the limit of this sequence. Let $y$ be a node in $\argmin_v v^*$ and assume $y^*<1$. Then take such a $y$ that has the closest distance from $s$. Since our model is a DeGroot model, among the neighbors of $y$, there is an $x\in N_\ell(y)$ with $x^*>y^*$, and all other neighbors $z$ of $y$ have $z^*\geq y^*$. This means that the sequence $v(\ell),v(\ell+1),\ldots$ will eventually surpass $y^*$.
\end{proof}
\begin{corollary}
If a decentralized DeGroot model stabilizes, then one access unit suffices for  opinion convergence. \qed
\end{corollary}

When the network does not stabilize, the convergence theorem no longer holds.
For example, when each time  $t$ involves two (or more) nodes entering $V_t$, it is  impossible to guarantee the convergence of opinions. 
Nevertheless, when we assume that the network $\g=\{G_t\}_{t\in \N}$ has {\em unit-growth}, that is, $V_{t-1}$ expands
by at most one node $v\in U\setminus V_{t-1}$ at each time $t$, it is possible to guarantee convergence of opinions.
In this case, a {\em dominating set}, a key notion in graph theory, plays a crucial role. Recall now the parameter $r$ from Section~\ref{DeGroot}.
The set $S_t$ is called {\em (distance-$r$) dominating} for network instance $G_t$ if $\Access^{(r)}_t(S_t)=V_t$. 
\begin{definition}
 A prowling algorithm is {\em domination-driven} if for any unit-growth decentralized network $\g$  there is a time
$t$ at which the algorithm produces $S_t$ that is dominating  for $G_t$.
\end{definition}
Algorithm~\ref{alg:framework} is a  centralization process that makes use of a domination-driven prowling algorithm. Note that the network evolution (Lines 4 \& 9) is beyond our control. 
The {\bf while}-loop repeatedly executes the $\prowl$ algorithm until $S_t$ dominates.
From then on, the algorithm will select new access units whenever $V_t$ expands.
\begin{definition}
A decentralized network $\g=\{G_t\}_{t\in \N}$ is {\em consistent} with $\prowl$ if each $S_t$, $t\in \N$,  is produced by Algorithm~\ref{alg:framework}. 
\end{definition}

\begin{algorithm}[h]
 \caption{A domination-driven centralization process}\label{alg:framework}
 \begin{algorithmic}[1]
 \item[PARAMETERS] radius $r$, num of steps $k$ of prowling
 \State Decentralized network starts at $G_0$ at time $t\leftarrow 0$
 \State $S_0\leftarrow \varnothing$
 \While{$V_t\setminus \Access_{t}(S_t) \neq \varnothing$}
    \State The network evolves to $G_{t+1}$ and set $t\leftarrow t+1$
    \State $v\leftarrow \prowl(S_{t-1})$ \Comment $k$-step prowl
    \State $S_{t}\leftarrow S_{t-1}\cup \{v\}$
 \EndWhile
 \Repeat
    \State The network evolves to $G_{t+1}$ and set $t\leftarrow t+1$
    \If{$V_t\setminus \Access_t(S_{t-1})\neq \varnothing$} \Comment A new node enters $V_t$
        \State Select $u$ from $V_t\setminus \Access_t(S_{t-1})$ \Comment $u$ is the only new node
        \State $S_t\leftarrow S_{t-1}\cup \{u\}$
    \EndIf
 \Until{Termination} \Comment{The loop runs infinitely as $G_t$ evolves.}
 \end{algorithmic}
 \end{algorithm}

\begin{theorem}[Domination theorem]\label{thm:converge-unit}
If decentralized network $\g$ is consistent with a domination-driven prowl and $\g$ has the unit-growth property, then
in $\g$ opinions converge. 
\end{theorem}
\begin{proof} Suppose $S_\ell$ is dominating to $G_\ell$ at time $\ell>0$.
At the first time $t>\ell$ when there is a node $v\in V_{t}\setminus \Access_{t}(S_{t-1})$, the algorithm selects the only node $v$ from $V_{t}\setminus \Access_{t}(S_{t-1})$. 
Thus $S_t=S_{t-1}\cup \{v\}$ is dominating to $G_t$. Repeating this procedure, one can prove inductively that for any $t'>\ell$, $S_{t'}$ forms a dominating set to $G_{t'}$ that contains all nodes in $V_{t'}\setminus V_t$. Using a similar argument as in Theorem~\ref{Thm:converge}, we can show that opinions in the sub-network restricted to nodes in $V_\ell$ converge, and hence the theorem holds.
\end{proof}

The {\em domination cost $\delta(\g)$} of a domination-driven $\prowl$ on a unit-growth network $\g=\{G_t\}_{t\in \N}$ consistent with $\prowl$ is the least $t$ at which $S_t$ dominates.  The theorem above links
the centralization and the domination costs: For any $0<\epsilon<1$, let $t_\epsilon$ be the least number of steps for all agents in $G_{\delta(\g)}=\left(V_{\delta(\g)}, E_{\delta(\g)}, C_{\delta(\g)}, S_{\delta(\g)}\right)$ to reach an opinion value of at least $1-\epsilon$. The $(1-\epsilon)$-centralization cost is then no more than $\delta(\g)+t_{\epsilon}$. We will verify in our experiments that the value $t_{\epsilon}$ is negligible (close to 0) even for very small $\epsilon$.

\section{The prowling algorithms}\label{S:Prowling}
In the rest of the paper, we assume that the network has unit growth, and aim to look for a domination-driven prowling algorithm with the smallest domination cost.

\subsection{Exterior-based Prowling}

We exhibit a simple condition where a prowling algorithm is domination-driven.
Define the {\em exterior set} at time $t$ as $\Ext_t=V_t\setminus \Access_t(S_{t-1})$.

\begin{definition}
We say that a prowling algorithm $\prowl(S_{t-1})$ is  {\em exterior-based} if the algorithm returns a node $u_t\in \Ext_t$.
\end{definition}

We aim to prove that exterior-based prowling algorithms are domination-driven if the radius $r\geq 2$. Assume that at each time $t>0$, a new node $v_t\notin \Access_t(S_{t-1})$ enters $V_t$, i.e., $V_{t}=V_{t-1}\cup \{v_t\}$.
%
We need the following.  At time $t$, the {\em access neighborhood} is $N_t=\bigcup_{s\in S_t} N_t(s)$, 
and  {\em anti-neighborhood} is $\Anti_t=V_t\setminus N_t$. We have $N_{t-1}\subseteq N_t$ and $\Anti_t\subseteq \Anti_{t-1}\cup \{v_t\}$.
Note also that $\Anti_t=\varnothing$ implies that $N_t=V_t$, i.e., $S_t$ is a (distance-1) dominating set.

\begin{lemma}\label{lem:neighbor}
For any $t>0$, $|\Anti_{t-1}| - |\Anti_t| \geq   |N_t\setminus N_{t-1}|-1$.
\end{lemma}
\begin{proof}
We have the following calculations:

\begin{align}
|\Anti_{t-1}|-|\Anti_t| &\geq |\Anti_{t-1}\cup\{v_t\}|-|\Anti_t|-1 \notag\\
&= |\Anti_{t-1}\!\cup\!\{v_t\}|\!-\!|(\Anti_{t-1}\cup\{v_t\})\!\cap\!\Anti_t|\!-\!1 \label{eqn:anti1}\\
&= |(\Anti_{t-1}\cup\{v_t\})\setminus \Anti_t|-1 \notag\\
 &= |(V_t\setminus N_{t-1}) \setminus (V_t\setminus N_t)|-1 \label{eqn:anti2}\\
 &=|(V_t\setminus (V_t\setminus N_t)) \setminus (N_{t-1}\setminus (V_t\setminus N_t))|-1 \notag\\
 &=|N_t\setminus (N_{t-1}\cap N_t)|-1=|N_t\setminus N_{t-1}|-1 \notag
\end{align}
where \eqref{eqn:anti1} is derived from $\Anti_t=(\Anti_{t-1}\cup\{v_t\})\cap \Anti_t$ and \eqref{eqn:anti2} from $\Anti_{t-1}\cup \{v_t\} = V_t\setminus N_{t-1}$.
\end{proof}

\begin{theorem}\label{thm:uset}
Any exterior-based prowling algorithm with the radius $r\geq 2$ is domination-driven. Moreover, the domination cost is no more than $t+|\Anti_t|$ for any $t>0$.
\end{theorem}

\begin{proof}
Let $u_t$ be the node returned by the prowling algorithm $\prowl(S_{t-1})$ when it is executed at time $t$.
By assumption $u_t\in \Ext_t$. Take a shortest path from $u_t$ to a node $s\in S_{t-1}$ and let $w$ be the node adjacent to $u_t$ on that path. It is clear that $w\notin \Access_t(S_{t-1})$ and $w\in \Access_t(\{u_t\})$. This means that $\{u,w\}\subseteq N_t$ while $\{u,w\}\cap N_{t-1}=\varnothing$. Hence $|N_t\setminus N_{t-1}|\geq 2$. By Lemma~\ref{lem:neighbor}, $|\Anti_{t-1}| - |\Anti_t| \geq   |N_t\setminus N_{t-1}|-1=1$ and thus $|\Anti_{t}|< |\Anti_{t-1}|$. Therefore for some $t'>0$, $\Anti_{t'}=\varnothing$ which means that $S_{t'}$ is dominating. Moreover, at any time $t$, the number of times before the anti-neighborhood becomes $\varnothing$ is at most $|\Anti_t|$. Hence the theorem is proved.
\end{proof}

By Theorem~\ref{thm:uset}, a centralization process consistent with an exterior-based prowling algorithm (with $r\geq 2$) 
guarantees to produce a dominating set over a unit-growth decentralized network. Furthermore, the theorem also provides a sequence of upper bounds on the domination cost, i.e.,
\[
1+|\Anti_1|, 2+|\Anti_2|, 3+|\Anti_3|,\ldots
\]
To lower domination cost, it makes sense to minimize values in this sequence. As $V_t$ grows at a constant rate, this is equivalent to making the access neighborhood $N_1, N_2, N_3, \ldots$ grow as fast as possible. This naturally suggests selecting access units based on their {\em degrees}. We thus design the following  exterior-based algorithm:

\begin{definition} We define $\degr(S_{t-1})$ the $k$-step {\em degree prowling algorithm}: Start from a random node $w_0$ in $\Access_t(S_{t-1})$ that has an edge going out of $\Access_t(S_{t-1})$, and move to a random neighbor $w_1\notin \Access_t(S_{t-1})$.  From $w_i$ ($1\leq i<k$), the algorithm selects among neighbors outside of $\Access_{t}(S_{t-1})$ one that (1) has a higher degree than $w_i$, and (2) the degree is as high as possible, i.e., set   $w_{i+1}=\argmax\{\degree(v)\mid v\in \Ext_t, \degree(v)>\degree(w_i)\}$. This step is repeated for at most $k-1$ steps or until it is not possible to do so within $k$ steps.
\end{definition}
The following directly follows from Thm.~\ref{thm:uset} and Thm.~\ref{thm:converge-unit}.
\begin{corollary}
If a unit-growth decentralized network $\g$ is consistent with  $\degr(S_{t-1})$ with $r\geq 2$, then in $\g$ opinions converge. \qed
\end{corollary}

\subsection{B-set Prowling}

The major downside of the exterior-based algorithms is that prowling is confined to the induced subgraph of the exterior set $\Ext_t$, which forbids more central nodes from being selected. This is problematic, especially after several access units have been selected, when the exterior subgraph is fragmented and contains many isolated nodes. In such a case, the algorithm may have a high domination cost. Fig.~\ref{fig:example} shows an example where an exterior-based algorithm performs poorly. Thus it is important to enable prowling to move {\em within the accessible region}.


\begin{figure}[ht!]
  \centering
  \includegraphics[width=0.5\linewidth]{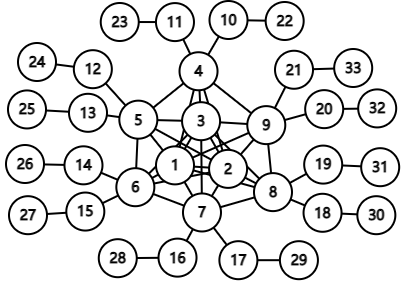}
  \caption{A static network with 33 nodes and 51 edges 
  with $r=2$, the optimal solution has 6 nodes, e.g., $\{4,5,6,7,8,9\}$. W.l.o.g, assume node $4$ is the initial node for all algorithms. Then
  the exterior-based algorithms outputs 11 nodes, e.g., $\{4,14,15,16,17,18,19,24,25,32,33\}$.}
  \label{fig:example}
\end{figure}

We therefore define {\em B-set prowling algorithms}. Recall that $r$ is the radius of the accessible region and set $r\dotdiv 1=r-1$ if $r\geq 2$ and $r\dotdiv 1=1$ if $r=1$.  For any time $t>0$, define the set
\begin{equation}\label{eqn:Bt}
B_t = \{v\in V_t\mid \dist(v,\Ext_t)\leq r\dotdiv 1\}.
\end{equation}
\begin{definition}
A {\em B-set prowling algorithm} $\prowl(S_{t-1})$ is one that returns a node $u_t\in B_t$.
\end{definition}

\begin{theorem}\label{thm:rset}
If $r\geq 2$, then a B-set prowling algorithm is domination-driven. Moreover, the domination cost is no more than $t+|\Anti_t|$ for any $t>0$.
\end{theorem}

\begin{proof}
Suppose $r\geq 2$. Write $u_t$ for the node returned by the B-set $\prowl(S_{t-1})$. By \eqref{eqn:Bt}, there is a node $w\in \Ext_t$ such that $\dist_t(u_t,w)\leq r-1$. 

We first discuss the case where $u_t=w$. In this case, $u_t\in \Ext_t$, implying that $u_t\notin N_{t-1}$. Notice that $\dist_t(u_t,S_{t-1})>r$, meaning that for any a node $s\in S_{t-1}$, $\dist_t(u_t,s)>r$. Take a shortest path between $s$ and $u_t$, and let $x$ be the node on this path that is adjacent to $u_t$. By triangle inequality, $\dist_t(x,s)\geq \dist_t(s,u_t)-1\geq r$. Note that $s$ is any a node in $S_{t-1}$, thus we have $\dist_t(u_t,S_{t-1})\geq r$, meaning that $x\notin N_{t-1}$. On the other hand, $u_t\in S_t$, implying that $\{u_t,x\}\subseteq N_t$. Thus $|N_t\setminus N_{t-1}|\geq 2$.

We then consider the case where $u_t\neq w$. In this case, we take a shortest path between $w$ and $u_t$ and let $y$ be the node on this path that is adjacent to $u_t$. Note that $y$ may be the same as $w$. 
Since $\dist_t(w,S_{t-1})> r$, we have  $\dist_t(v,w)>r-1\geq 1$ for any $v\in N_{t-1}$. Hence as $\dist_t(u_t,w)\leq r-1$, $u_t\notin N_{t-1}$. Also, $\dist_t(y,w)\leq r-2$ implies that $y\notin N_{t-1}$.  On the other hand, $u_t\in S_t$, implying that $\{u_t,y\}\in N_t$. Thus $|N_t\setminus N_{t-1}|\geq 2$.

In either case, $|N_t\setminus N_{t-1}|\geq 2$. By Lem.~\ref{lem:neighbor}, $|\Anti_t|\!<\!|\Anti_{t-1}|$. Thus in a similar argument as in the proof of Thm.~\ref{thm:uset}, there is a time $t'>0$ where $\Anti_{t'}=\varnothing$ and $S_{t'}$ is dominating. Moreover, the dominating cost can be analyzed in the same way as in Theorem~\ref{thm:uset}.
\end{proof}

We now define $\boomerang(S_{t-1})$, a $k$-step prowling algorithm.
Recall that $k$ is the maximum number of steps in each call to prowling. 

\begin{definition} The $k$-step {\em B-degree algorithm} $\boomerang(S_{t-1})$ starts from a random node $w_0$ in $\Access_t(S_{t-1})$ that has an edge going out of $\Access_t(S_{t-1})$, and moves to a random neighbor $w_1\notin \Access_t(S_{t-1})$.  From $w_i$ ($1\leq i<k$), the algorithm selects among neighbors in $B_t$ one that (1) has a higher degree than $w_i$, and (2) the degree is as high as possible, i.e., set   $w_{i+1}=\argmax\{\degree(v)\mid v\in B_t, \degree(v)>\degree(w_i)\}$. This step is repeated for at most $k-1$ steps or until it is not possible to do so within $k-1$ steps.
\end{definition}

Note that the upper bounds of the dominating costs, as asserted by Theorem~\ref{thm:rset}, imply that it is preferred to make the access neighborhood $N_1, N_2, N_3, \ldots$ grow as fast as possible. Maximizing the degrees of the selected nodes does not optimally speed up the growth of $N_t$; Instead we define {\em N-degree} of a node $v$ as the number $\degN_t(v)$ of its neighbors that are not in $N_t$. 

\begin{definition} The $k$-step {\em B-Ndegree algorithm} $\bndeg(S_{t-1})$ starts from a random node $w_0$ in $\Access_t(S_{t-1})$ that has an edge going out of $\Access_t(S_{t-1})$, and moves to a random neighbor $w_1\notin \Access_t(S_{t-1})$.  From $w_i$ ($1\leq i<k$),  among neighbors of $w_i$ in $B_t$, do: 
\begin{itemize}[leftmargin=*]
\item If $w_i\in \Access_t(S_{t-1})$, then set $w_{i+1}$ as the neighbor of $w_i$ that has a higher N-degree and $\degN_t(w_{i+1})$ is maximized. 
\item if $w_i\notin \Access_t(S_{t-1})$ (note that in this case, the N-degree of neighbors of $w_i$ is not accessible), then set $w_{i+1}$ as the neighbor of $w_i$ that has a higher degree and $\degree_t(w_{i+1})$ is maximized. 
\end{itemize}
This step is repeated for at most $k-1$ steps or until it is not possible to do so within $k-1$ steps.
\end{definition}
The following directly follows from Thm.~\ref{thm:rset} and Thm.~\ref{thm:converge-unit}.

\begin{corollary}
If a  unit-growth decentralized network $\g$ is consistent with $\boomerang(S_{t-1})$ or $\bndeg(S_{t-1})$ with $r\geq 2$, then in $\g$ opinions converge. \qed
\end{corollary}

\section{Experiments}\label{sec:experiment}
\subsection{Setup}
We perform a series of experiments to validate the performance of our prowling algorithms. 
These experiments serve the following purposes. The first is to test the convergence of the algorithms over real-world dynamic networks. The second is to compare the domination costs of different prowling algorithms. The third is to test the effectiveness of the algorithms in facilitating opinion convergence. 

\paragraph*{\bf Algorithms.} We run Algorithm~\ref{alg:framework} using the three  proposed prowling algorithms: $\degr$, $\boomerang$, and $\bndeg$. To put them in perspective, we also test the following benchmarks:
\begin{itemize}[leftmargin=*]
\item {\bf Local random algorithm ($\lrand$):} This is a local  algorithm that picks a random node adjacent to $\Access_t(S_{t-1})$ as the next access unit. It does not use prowling. 

\item {\bf Local degree algorithm ($\local$):} This local algorithm resembles the one proposed in \cite{borgs2012power}. When selecting the $t$th access unit, it picks the node $u\in \Access_t(S_{t-1})$ that has the highest N-degree. This, just like $\lrand$, also does not use prowling.

\item {\bf Exterior random algorithm ($\random$):} This prowling algorithm starts by moving out of $\Access_t(S_t)$ and performs a random walk (of length $k-1$) in the exterior set $X_t$ to choose the next access unit. In this way, the algorithm can be considered exterior-based, and thus  Thm~\ref{thm:uset} asserts it is domination-driven. 

\item {\bf B-set random algorithm ($\brandom$):} This prowling algorithm is defined in the same way as $\random$ except that the random walk takes place in $B_t$ rather than $X_t$. In this way, the algorithm is a B-set prowling algorithm, and by Thm.~\ref{thm:rset} it is domination-driven.
\end{itemize}

\paragraph*{\bf Real-world datasets.} We utilize four real-world datasets in our experiments. 
{\em Facebook-WOSN (Facebook)} records a social network on Facebook with time stamps \cite{viswanath-2009-activity}. An edge $\{u,v\}$ denotes a friendship relation between two users. {\em  Wiki-talk (Wikitalk)} is a timestamped communication network of Basque Wikipedia.
An edge $\{u,v\}$ denotes that user $u$ wrote a message on the talk page of user $v$. {\em Citation-HepPh (Citation)} is a high-energy physics citation network \cite{leskovec2007graph}, where an edge $\{u,v\}$ denotes a citation relationship between papers $u$ and $v$. 
{\em Enron-Email (Enron)} is a timestamped email network between the
employees of a company \cite{klimt2004enron} where an edge $\{u,v\}$ denotes emails between employees $u$ and $v$.
Table~\ref{tab:dataset} shows key statistics of these datasets. 
Note that these datasets involve only incremental changes of the networks. To ensure that the network has a sufficient size as well as a sufficient number of time stamps, we begin the centralization process from different initial time stamps. Moreover, to approximate unit growth, we run the prowling algorithm (thereby selecting an access unit) for every 13,1,1,2 times stamps for Facebook, Wikitalk, Citation, and Enron, respectively.

\begin{table}[t]
  \centering \footnotesize
  \caption{\footnotesize Key statistics of the real-world networks. The last six rows are properties about the last time stamp. A high {\em clustering coefficient} and low {\em diameter} indicate a small-world network.}\label{tab:dataset}
  \begin{tabular}{|l|c|c|c|c|}
  \hline  &Facebook&Wikitalk&Citation&Enron \\ \hline
  \hline \# of time stamps &876,992&55,197&3,308&220,363 \\
  \hline initial time stamp & 400,000& 27,000 & 375 & 100,000 \\ 
  \hline initial $|V|$ &21,059&15,807&14,526&45,596  \\
  \hline initial $|E|$  &81,473 &17,765&30,815&156,596 \\
  \hline final $|V|$ &63,731&40,993&34,449&87,273  \\
  \hline final $|E|$  &193,493&45,682&405,823&299,219 \\
  \hline average degree &8.24&2.23&23.56&6.86  \\
  \hline clustering coefficient &0.11&0.02&0.28&0.12  \\
  \hline max degree &223&29,479&846&1,728 \\
  \hline diameter &15&7&15&14 \\
\hline
  \end{tabular}
\end{table}

\paragraph*{\bf Dynamic models.} We also run four models of dynamic networks. The goal is to test the performance and robustness of our algorithms over networks with different structural properties. {\em Dynamic Barabasi-Albert model (BA).} This produces standard scale-free networks through preferential attachment \cite{barabasi1999emergence}. The model adds a new node at each time stamp which randomly connects with $\overline d/2$ nodes, where $\overline d$ is the average degree. The network instances will exhibit a power-law degree distribution.   {\em Dynamic stochastic block model (SB)} is well-used for generating a community structure \cite{holland1983stochastic}. We utilize this model to generate networks with three equal-size communities. 
At each time stamp, a new node enters a random community. It then links to nodes in the same community with a high probability $p_{in}$, and to nodes in other communities with a low probability $c p_{in}$ with $c=1/6$ to mimic real-world networks as shown in \cite{holland1983stochastic}. 
{\em Dynamic Jackson-Rogers model (JR).} This model, proposed by \cite{jackson2007meeting}, simulates stochastic friendship-making among individuals in a population. The network instances meet both scale-free and small-world properties. The dynamic version of the model was introduced in \cite{ijcai2018-544}. There are two parameters $p$ and $\overline{d}$. As in \cite{ijcai2018-544}, we pick $p=0.5$. 
{\em Dynamic rich-club model (RC).} The model develops a typical core/periphery structure with a dense and central core with a sparse periphery \cite{bornholdt2001world,csermely2013structure}. At each timestamp, the model adds a new node with probability $\alpha\in[0,1]$ or an edge between two existing nodes with probability $1-\alpha$, where $\alpha\!=\!2\!/\!\overline d$. The dynamic version of this model was introduced in \cite{ijcai2018-544}. 
For easier comparison, we classify networks generated by these models using two parameters: the initial size $N$ and the average degree $\overline{d}$. 
Table~\ref{tab:models} summarizes key statistics of the models by setting the average degree\footnote{80 datasets on KONECT and SNAP have average degrees between 2 and 10 \url{http://konect.uni-koblenz.de/}, \url{http://snap.stanford.edu/}} $\overline d=6$ and network size $N=5000$. For these models, there is no upper bound on the number of time stamps that we can generate.

\begin{table}[t]
  \centering \footnotesize
  \caption{\footnotesize Key statistics of dynamic network models with $\overline d=6$ and $N=5000$}\label{tab:models}
  \begin{tabular}{|l|c|c|c|c|c|c|c}
\hline  &BA&SB&JR&RC\\ \hline 
 \hline clustering coefficient&0.01&0.0015&0.34&0.01   \\
 \hline max degree &202&16&52 &107 \\
 \hline diameter &8&12&10&13\\
\hline
\end{tabular}
\end{table}

\paragraph*{\bf Experiment design.} We perform three experiments. 
Exp.~1 compares the domination (and centralization) costs of different prowling algorithms over real-world networks. For this, we vary the radius $r$ and steps $k$ of prowling. For each dataset, we  start from the specified initial time stamp with no accessing units, running prowling until the access units form a dominating set. At any time step, we assume that $G_{t+1}=G{t}$ if no successive timestamp is provided in the dataset.    
Exp.~2 compares the algorithms over networks generated by dynamic network models. Here, we evaluate the domination costs with different initial network size $N$ and average degree $\overline{d}$. 
Exp.~3 further validates the findings from the two experiments above by evaluating the opinion values as a result of different prowling algorithms. We test the average opinion values of nodes in $G_t$ during a centralization process until this value becomes higher than 0.99. An algorithm that results in faster growth in opinion values can be considered more effective at centralization.

\subsection{Results and discussion}
\paragraph*{\bf Experiment 1.} When running on the real-world datasets with the settings above, all prowling algorithms produce a dominating set, and hence opinions converge in the network. Our discussions above point out a correlation between domination cost and centralization cost, namely, the $(1-\epsilon)$-centralization cost equals to the domination cost plus $t_\epsilon$ (see Section~\ref{sec:unit}). In this experiment, we compute the value of $t_\epsilon$ when $\epsilon$ is roughly 0. It turns out that, over the real-world networks, for all cases considered, $t_\epsilon=0$ for $\epsilon\sim 0$. This validates our conjecture above that domination costs directly reflects centralization costs. 

Table~\ref{tab:real-steps} shows the domination costs for $r\in \{1,2\}$ and $k\in \{4,7,10\}$ after 10 runs of each algorithm. We make the following observations: First, when comparing the local algorithms ($\lrand$, $\local$) with the prowling algorithms, the prowling algorithms in general achieve lower domination costs. This is especially so for $r=2$ which is consistent with the theoretical analysis above. For $r=1$, the $\local$ algorithm tends to outperform $\lrand$ and the exterior-based algorithms. Nevertheless, the B-set prowling algorithms in general outperform the other algorithms in all cases. This indicates that prowling could result in much smaller centralization costs than local algorithms that do not use prowling. Moreover, B-set prowling in general is more suitable than exterior-based prowling. Second, the prowling algorithms that are based on degrees ($\degr$, $\boomerang$, $\bndeg$) in general perform better than their random-walk counterparts. In particular, $\degr$ clearly outperforms the other exterior-based algorithm $\random$ in all cases; $\boomerang, \bndeg$ clearly outperform $\brandom$ (with the Wikitalk $r=1,k=7$ as the only exception). This is also consistent with our theoretical analysis as we stipulate that degrees can be used as an effective criterion for centralization. Third, between the two degree-based B-set prowling algorithms, we observe comparable performance when $r=2$. Nevertheless, $\bndeg$ performs significantly better when $r=1$ for all datasets apart from Wikitalk. This suggests that it may be more appropriate to adopt N-degree in prowling. 

A slight unexpected result is that when $k$ increases, there does not seem to be a noticeable difference on the domination cost. This may be due to the fact that the degree ascending paths in the network are in general very short, or the exterior set (and B-set) are fragmented as access units are selected, resulting in a lack of long paths. In all subsequent experiments, we thus set $k=4$. 

Also, Wikitalk demonstrates noticeable differences as compared to the other three datasets. This can be explained by the distinct topology of this network. Indeed, Wikitalk is very sparse  with a  much lower average degree and clustering coefficient than others, and it holds a node with a dis-proportionally high degree (29,479).

\begin{table*}[ht!]
\centering
\caption{Experiment 1. The domination costs of different prowling algorithms over real-world datasets for $r\in \{1,2\}$ and $k\in \{4,7,10\}$.}
\label{tab:real-steps}
\resizebox{0.8\linewidth}{!}{%
\begin{tabular}{|c|c|l|l|l|l|l|l|l|l|l|l|l|}
\hline
 & \multicolumn{6}{c|}{Facebook} & \multicolumn{6}{c|}{Wikitalk} \\ \hline
r & \multicolumn{3}{c|}{1} & \multicolumn{3}{c|}{2} & \multicolumn{3}{c|}{1} & \multicolumn{3}{c|}{2} \\ \hline
$k$ & 4 & \multicolumn{1}{c|}{7} & \multicolumn{1}{c|}{10} & \multicolumn{1}{c|}{4} & \multicolumn{1}{c|}{7} & \multicolumn{1}{c|}{10} & \multicolumn{1}{c|}{4} & \multicolumn{1}{c|}{7} & \multicolumn{1}{c|}{10} & \multicolumn{1}{c|}{4} & \multicolumn{1}{c|}{7} & \multicolumn{1}{c|}{10} \\ \hline
$\lrand$ & \multicolumn{3}{c|}{8,490${\pm 80}$} & \multicolumn{3}{c|}{3,831${\pm 39}$} & \multicolumn{3}{c|}{14,370${\pm 4,090}$} & \multicolumn{3}{c|}{187${\pm 9}$} \\ \hline
$\local$ & \multicolumn{3}{c|}{7,868${\pm 19}$} & \multicolumn{3}{c|}{5,579${\pm 16}$} & \multicolumn{3}{c|}{196${\pm 0}$} & \multicolumn{3}{c|}{184${\pm 0}$} \\ \hline
$\random$ & \multicolumn{1}{l|}{8,333${\pm 64}$} & 8,328${\pm 42}$ & 8,322${\pm 49}$ &3,866${\pm 33}$ &3,831${\pm 22}$ &3,854${\pm 18}$ &6,890${\pm 6,119}$ &13,372${\pm 6,024}$ &12,183${\pm 4,806}$ &185${\pm 20}$ &192${\pm 17}$ &181${\pm 21}$ \\ \hline
$\degr$ & \multicolumn{1}{l|}{7,619${\pm 52}$} & 7,606${\pm 42}$ & 7,586${\pm 19}$ &3,710${\pm 51}$ &3,692${\pm 26}$ &3,725${\pm 44}$ &2,701${\pm 24}$ &2,713${\pm 5}$ &2,712${\pm 10}$ &185${\pm 7}$ &185${\pm 7}$ &183${\pm 8}$ \\ \hline
$\brandom$ & \multicolumn{1}{l|}{7,256${\pm 102}$} & 7,560${\pm 91}$ & 7,328${\pm 63}$ & 3,358${\pm 33}$ & 3,470${\pm 15}$ &3,388${\pm 23}$ &204${\pm 4}$ & \textbf{173 ${\pm 10}$} &210${\pm 4}$ &155${\pm 5}$ &171${\pm 10}$ &157${\pm 7}$ \\ \hline
$\boomerang$ & \multicolumn{1}{l|}{6,926${\pm 29}$} & 6,943${\pm 32}$ & 6,916${\pm 40}$ & \textbf{3,210${\pm 16}$} & \textbf{3,205${\pm 30}$} &3,220${\pm 27}$ & \textbf{194${\pm 0}$} &194${\pm 0}$ & \textbf{194${\pm 0}$} & \textbf{148${\pm 0}$} & \textbf{148${\pm 0}$} & \textbf{148${\pm 0}$} \\ \hline
$\bndeg$ & \multicolumn{1}{l|}{\textbf{6,360${\pm 14}$}} & \textbf{6,348${\pm 26}$} & \textbf{6,353${\pm 13}$} &3,225${\pm 8}$ &3,210${\pm 20}$ & \textbf{3,218${\pm 45}$} & \textbf{194${\pm 0}$} &194${\pm 0}$ & \textbf{194${\pm 0}$} & \textbf{148${\pm 0}$} & \textbf{148${\pm 0}$} & \textbf{148${\pm 0}$} \\ \hline \hline
\multicolumn{1}{|l|}{} & \multicolumn{6}{c|}{Citation} & \multicolumn{6}{c|}{Enron} \\ \hline
r & \multicolumn{3}{c|}{1} & \multicolumn{3}{c|}{2} & \multicolumn{3}{c|}{1} & \multicolumn{3}{c|}{2} \\ \hline
$k$ & 4 & \multicolumn{1}{c|}{7} & \multicolumn{1}{c|}{10} & \multicolumn{1}{c|}{4} & \multicolumn{1}{c|}{7} & \multicolumn{1}{c|}{10} & \multicolumn{1}{c|}{4} & \multicolumn{1}{c|}{7} & \multicolumn{1}{c|}{10} & \multicolumn{1}{c|}{4} & \multicolumn{1}{c|}{7} & \multicolumn{1}{c|}{10} \\ \hline
$\lrand$ & \multicolumn{3}{c|}{2,457${\pm 45}$} & \multicolumn{3}{c|}{498${\pm 14}$} & \multicolumn{3}{c|}{21,313${\pm 12405}$} & \multicolumn{3}{c|}{824${\pm 32}$} \\ \hline
$\local$ & \multicolumn{3}{c|}{1,748${\pm 32}$} & \multicolumn{3}{c|}{870${\pm 10}$} & \multicolumn{3}{c|}{2,514${\pm 9}$} & \multicolumn{3}{c|}{2,059${\pm 11}$} \\ \hline
$\random$ &2,405${\pm 48}$ & \multicolumn{1}{c|}{2,471${\pm 40}$} & \multicolumn{1}{c|}{2416${\pm 31}$} & \multicolumn{1}{c|}{495${\pm 15}$} & \multicolumn{1}{c|}{485${\pm 12}$} & \multicolumn{1}{c|}{503${\pm 23}$} &19,697${\pm 175}$ &20,906${\pm 395}$ &19,593${\pm 940}$ &811${\pm 11}$ &811${\pm 25}$ &808${\pm 15}$ \\ \hline
$\degr$ & \multicolumn{1}{l|}{2,337${\pm 30}$} &2,290${\pm 57}$ &2,300${\pm 36}$ &455${\pm 24}$ &465${\pm 17}$ &463${\pm 23}$ &17,693${\pm 590}$ &17,684${\pm 145}$ &17,817${\pm 704}$ &777${\pm 12}$ &771${\pm 11}$ &780${\pm 23}$ \\ \hline
$\brandom$ & \multicolumn{1}{l|}{17,57${\pm 38}$} &2,105${\pm 59}$ &1,819${\pm 35}$ &338${\pm 19}$ &444${\pm 25}$ &359${\pm 16}$ &2,397${\pm 17}$ &12,457${\pm 311}$ &2,556${\pm 31}$ &741${\pm 36}$ &755${\pm 30}$ &720${\pm 28}$ \\ \hline
$\boomerang$ & \multicolumn{1}{l|}{1,596${\pm 22}$} &1,599${\pm 23}$ &1,595${\pm 26}$ &319${\pm 12}$ &321${\pm 10}$ & \textbf{316${\pm 20}$} &2,212${\pm 7}$ &2,210${\pm 11}$ &2,214${\pm 9}$ & \textbf{699${\pm 7}$} &699${\pm 13}$ & \textbf{697${\pm 4}$} \\ \hline
$\bndeg$ & \multicolumn{1}{l|}{\textbf{1,548${\pm 27}$}} & \textbf{1543${\pm 24}$} & \textbf{1,551${\pm 28}$} & \textbf{316${\pm 21}$} & \textbf{312${\pm 8}$} &320${\pm 17}$ & \textbf{2,137${\pm 5}$} & \textbf{2,135${\pm 10}$} & \textbf{2,133${\pm 10}$} &700${\pm 7}$ & \textbf{696${\pm 4}$} &700${\pm 9}$ \\ \hline
\end{tabular}%
}
\end{table*}

\smallskip

{\bf Experiment 2.} We perform three tests over networks generated by the dynamic models: BA, SBM, JR, and RC. 
(1) Setting $r=2$, and $N=5000$, we then compare domination costs with varying average degree $\overline{d}$.
(2) Finally, setting $r=2$, $\overline{d}=6$, we compare domination costs with varying initial size $N$. 
(3) Setting $N=5000$, and $\overline{d}=6$, we first compare domination costs with varying radius $r$. 
The results are shown in Fig.\ref{fig:syndeg}.
In general, we can see $\prowl_\boomerang$ and $\prowl_\bndeg$ outperform other algorithms.

As expected, when we increase the average degree from 4 to 14 or increase the radius from 1 to 4, all prowling algorithms achieve smaller domination costs, and moreover they tend to have similar performance. On the contrary, as we increase the initial size of the network, the domination costs also tend to increase. All algorithms exhibit rather stable performance across the cases where in general, $\bndeg$ and $\boomerang$ achieve the best performance. $\local$ results in far higher domination costs than the other algorithm in many cases. The results also reflect differences between the networks with different structural properties. In particular, $\local$ tends to perform well in the scale-free networks (BA), where as all algorithms apart from $\local$ perform comparably well in SBM.

\begin{figure}[ht]
  \centering
  \includegraphics[width=1\linewidth]{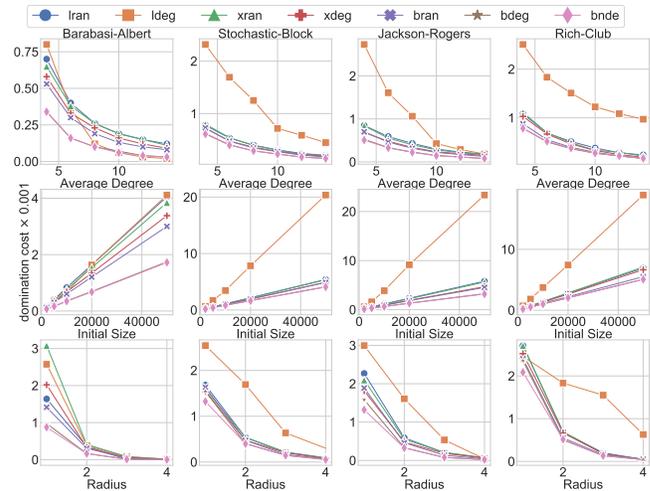}
  \caption{\footnotesize Experiment 2. Domination costs of synthetic networks produced by four dynamic network models (columns) with varying average degree (top row), initial size (middle row) and radius (bottom row). }\label{fig:syndeg}
\end{figure}

\paragraph*{\bf Experiment 3.} For further analysis, we calculate the average opinion values during the centralization processes over the real-world networks. 
Fig.~\ref{fig:realave} displays changes to the average opinion starting from the initial time stamp, as more access units a chosen by different prowling algorithms (fixing $r=2$), until the average opinion reaches above 0.8. In all datasets, the B-set prowling algorithms $\boomerang$ and $\bndeg$ are the ones that lift the average opinion the fastest, and are the first ones to achieve opinion value $\geq 0.8$. It is also remarkable that this is achieved only within a small number ($\leq 50$) of steps. In the citation network, $\local$ achieves comparable performance as the prowling algorithms, which may be explained by the high density of this dataset. In general, the degree-based prowling methods perform better than their random-walk counterparts. These findings are consistent with the two experiments above and provide further confirmation of the superior effectiveness of the degree-based B-set prowling algorithms in centralization.


\begin{figure}[ht!]
  \centering
  \includegraphics[width=\linewidth]{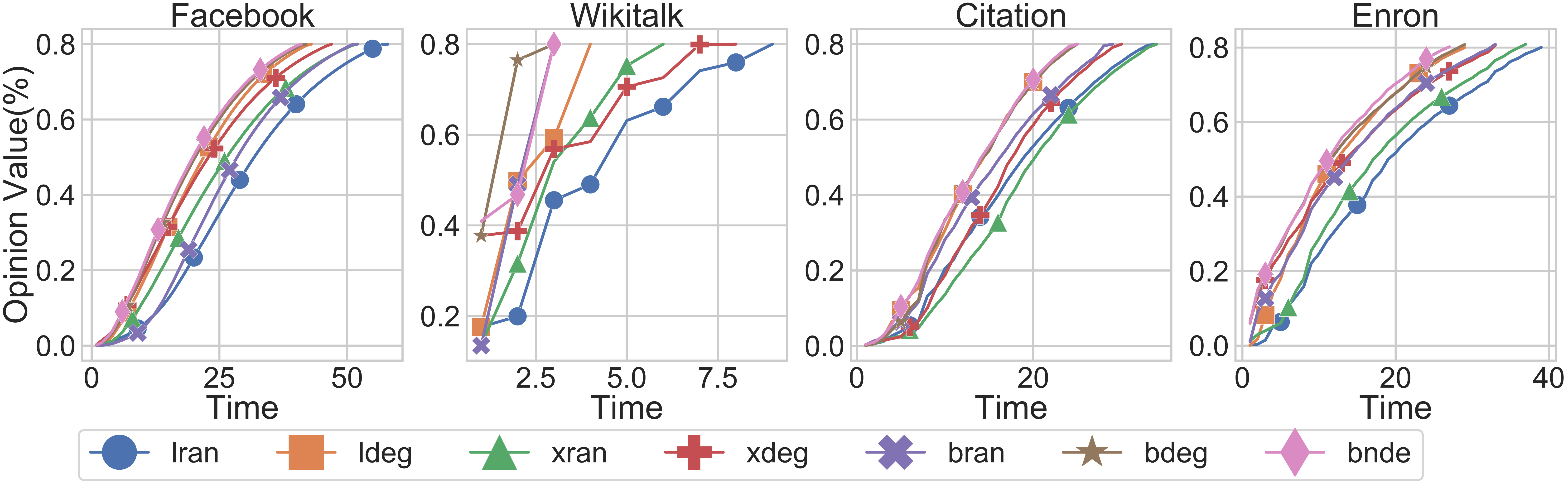}
  \caption{Experiment 3. Changes to the average opinion values with time over the real-world networks.}\label{fig:realave}
\end{figure}

\section{Conclusion and Future Work}
This paper investigates the interplay between decentralized systems and centralized control by introducing several novel concepts: (1) {\em Decentralized networks} are dynamic and unknown networks where agents are monitored and influenced through access units. (2) The {\em decentralized DeGroot model} extends the classical DeGroot model of opinion dynamics to a decentralized setting. (3) The {\em centralization problem} seeks a way to impose control to a decentralized network by selecting access units. (4) {\em prowling} are adaptive methods to explore the decentralized network. From a technical standpoint, we investigate prowling algorithms that facilitate opinion convergence. Firstly, we reveal the relationship between {\em domination} and {\em centralization} and show that domination-driven prowling algorithms hold the key to opinion convergence. Then, we explore prowling algorithms in unit-growth networks that guarantee to find dominating sets. Our theoretical finding inspires us to apply degree (or N-degree)-based prowling to reduce domination costs. Lastly, experimental results on real-world and synthetic dynamic networks verify that our algorithm achieves better domination costs compared to benchmarks. 

The framework proposed in this paper, i.e., decentralized networks and the centralization problem, has the potential to provide a new perspective on decentralized systems that are of interest in multi-agent systems, computational social sciences, and research about the Web. Further work may include investigating new local algorithms for selecting access units, as well as exploring centralization for other applications such as influence maximization (in a similar way as \cite{wilder2018maximizing}) or norm emergence. 

\bibliographystyle{ACM-Reference-Format}
\bibliography{aamas22}

\end{document}